\documentclass{article}

\usepackage[preprint]{neurips_2023}

\usepackage[utf8]{inputenc} 
\usepackage[T1]{fontenc}    
\usepackage{hyperref}       
\usepackage{url}            
\usepackage{booktabs}       
\usepackage{amsfonts}       
\usepackage{nicefrac}       
\usepackage{microtype}      
\usepackage{xcolor}         
\usepackage{subcaption}
\usepackage[linesnumbered,ruled]{algorithm2e}

\usepackage[utf8]{inputenc} 
\usepackage[T1]{fontenc}    
\usepackage{hyperref}       
\usepackage{url}            
\usepackage{booktabs}       
\usepackage{amsfonts}       
\usepackage{nicefrac}       
\usepackage{microtype}      
\usepackage{xcolor}         

\usepackage{amsmath}
\usepackage{amssymb}
\usepackage{mathtools}
\usepackage{amsthm}

\usepackage{thm-restate}

\usepackage{color-edits}
\addauthor{ev}{blue}

\theoremstyle{plain}
\newtheorem{theorem}{Theorem}[section]

\newtheorem{lemma}[theorem]{Lemma}

\theoremstyle{definition}
\newtheorem{definition}[theorem]{Definition}
\newtheorem{assumption}[theorem]{Assumption}
\theoremstyle{remark}

\newenvironment{hproof}{%
  \proof}{\endproof}

\title{No Bidding, No Regret: Pairwise-Feedback Mechanisms for Digital Goods and Data Auctions}

\author{%
  Zachary Robertson \\
  Department of Computer Science\\
  Stanford\\
  \texttt{zroberts@stanford.edu} \\
  \And
  Oluwasanmi Koyejo \\
  Department of Computer Science\\
  Stanford\\
  \texttt{sanmi@stanford.edu} \\
}

\begin{document}

\maketitle

\begin{abstract}
    The growing demand for data and AI-generated digital goods, such as personalized written content and artwork, necessitates effective pricing and feedback mechanisms that account for uncertain utility and costly production. Motivated by these developments, this study presents a novel mechanism design addressing a general repeated-auction setting where the utility derived from a sold good is revealed post-sale. The mechanism's novelty lies in using pairwise comparisons for eliciting information from the bidder, arguably easier for humans than assigning a numerical value. Our mechanism chooses allocations using an epsilon-greedy strategy and relies on pairwise comparisons between realized utility from allocated goods and an arbitrary value, avoiding the learning-to-bid problem explored in previous work. We prove this mechanism to be asymptotically truthful, individually rational, and welfare and revenue maximizing. The mechanism's relevance is broad, applying to any setting with made-to-order goods of variable quality. Experimental results on multi-label toxicity annotation data, an example of negative utilities, highlight how our proposed mechanism could enhance social welfare in data auctions. Overall, our focus on human factors contributes to the development of more human-aware and efficient mechanism design.
\end{abstract}

\section{Introduction}

Marketplaces generating digital goods, such as personalized written content and artwork based on user requests, have garnered significant attention in recent years due to their ability to scale and adapt to user preferences \citep{Morgan_2023, Paul_Dang_2023}. Such generative marketplaces possess immense potential to revolutionize the economy through applications such as online advertising \citep{Paul_Dang_2023}, a market where spending is expected to exceed
\$700 billion in 2023 \citep{Statista_2023}. However, they face challenges in collecting accurate and timely human feedback, as well as in managing compute costs for the most advanced models, which the CEO of OpenAI has described as "eye-watering" \citep{Karpf_2023a}. This problem is particularly acute since the value each user derives from a fulfilled request is typically only known after allocation.

In this paper, we examine the general repeated-auction setting where the utility derived from sold digital goods is revealed to bidders post-sale. In this setting, digital goods are "made-to-order" based on user requests. A key challenge is that users could provide inaccurate or misleading feedback, which would harm revenue generation. To address these challenges, we propose an auction mechanism that is robust to strategic reporting on the user side and no-regret in revenue on the market side. Our pricing mechanism is based on a pairwise comparison model that asks the user to report if the value of their allocation is above an arbitrarily selected reference point, avoiding the learning-to-bid problem that has been a point of concern in previous works \citep{feng2018learning, guo2022no}.

Our main contribution is a novel auction mechanism for selling digital goods that are costly to produce and whose utility to a particular user is uncertain. We outline our main contributions below:

\begin{enumerate}
\item \textbf{Feedback-based auction mechanisms:} This study introduces a feedback-driven, contextual, asymptotically truthful mechanism that eliminates the need for users to know their value for generating a digital good beforehand. By allocating goods to agents and subsequently collecting feedback on their satisfaction, the mechanism effectively sidesteps the learning-to-bid problem.

\item \textbf{Analysis of efficiency:} An in-depth analysis of the proposed mechanism is presented. We establish finite-time regret bounds for truthful reporting, participation, and welfare-revenue generation against the standard second-price auction. We also show the underlying expected utilities can be identified from pairwise comparisons without relying on distributional assumptions. 

\item \textbf{Welfare maximizing data acquisition:} We also explore how to use our mechanism as a payment rule for toxicity annotation, a setting with negative utilities that are only realized after the mechanism purchases a label from a user. We discuss this setting further in Section \ref{toxicity_setting}. 

\end{enumerate}

We tackle two main technical challenges. First, akin to \citep{nazerzadeh2013dynamic}, we utilize a learning algorithm for the expected utility function but diverge from their history-dependent rule, which requires $O(H^2)$ calls to the model for $H$ allocation rounds to enhance computational efficiency. Our Lemma \ref{strategicutility} outlines that providing inaccurate or misleading feedback isn't particularly profitable, achieved through a refined strategic reporting and regret analysis. Secondly, our mechanism features a simplified reporting rule, negating the need for agents to accurately learn to bid, as seen in works like \citep{feng2018learning, guo2022no}. Although common among all prior work we are aware of, we also dismiss the unrealistic assumption of precise value reporting as it essentially requires reporting a real number with infinite precision. In Theorem \ref{pairMech}, we establish a mechanism that uses feedback reports free from distributional assumptions on underlying utilities.

The remainder of this paper is organized as follows. In Section \ref{related_work}, we review related work in auctions with incomplete information, mechanism design, and data pricing. In Section \ref{background}, we formalize our setting and describe the background for our approach. In Section \ref{model} we introduce our proposed mechanism, and in Section \ref{analysis} we present our main results for the proposed mechanism. 

\begin{figure*}[!t]
    \centering
    ~ 
    \centering
    \includegraphics[width=0.8\linewidth]{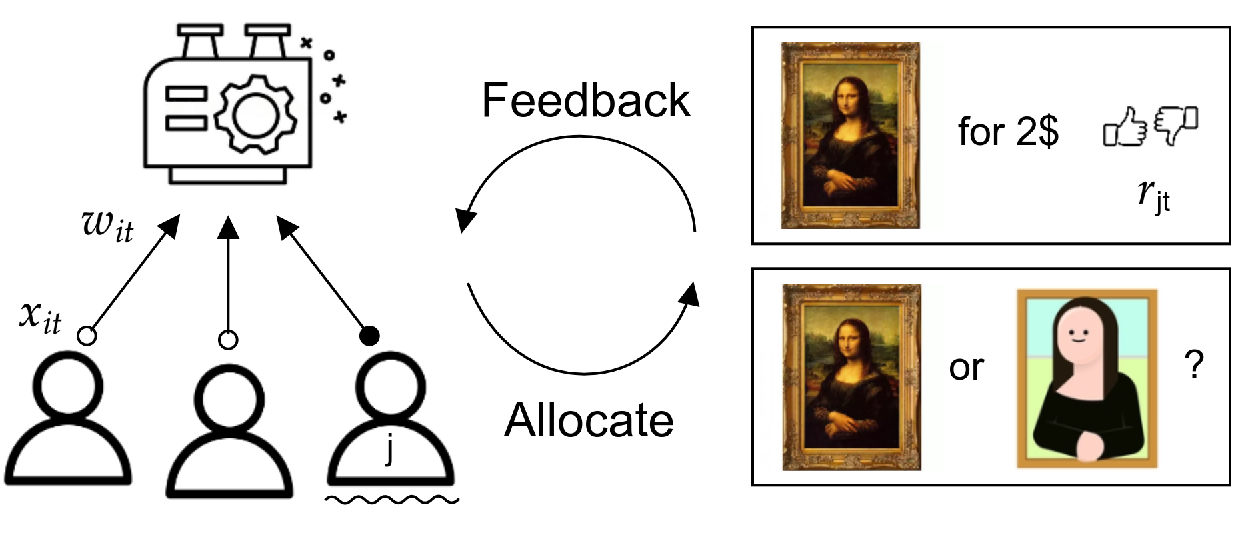}
    
    \caption{This is an illustration of our proposed mechanism introduced in Section \ref{model}. In the left panel we start period $t$, each agent $i \in [n]$ submits a request $w_{it}$ (prompt) for a made-to-order digital good, and then the mechanism determines a boolean-valued allocation assignment $x_{it}$ for each agent. If the $j^{\text{th}}$ agent receives an allocation for their request ("Mona Lisa") then the agent receives a digital good made according to their request. After a digital good is produced, the agent self-reports a Boolean value, denoted by $r_{jt}$, indicating if the value of their allocation is above an arbitrarily selected reference point. Alternatively, the mechanism produces multiple digital goods and has the user rank them. A priori, it is not obvious how to design a mechanism robust to strategic manipulation. This paper devises a simple mechanism based on a second-price auction that is asymptotically truthful.}
    \label{fig:Motivation}
\end{figure*}

\section{Related work}
\label{related_work}

Our work builds on research in the fields of auctions with incomplete information, machine learning for mechanism design, and data pricing. 

\subsection{Mechanism design and machine learning}

Our work intertwines mechanism design and machine learning to address challenges in pricing and feedback systems for digital goods and data \citep{balcan2005mechanism, devanur2009price, babaioff2009characterizing, babaioff2015truthful}. We also draw from the intersection of machine learning and mechanism design for allocating digital goods \citep{immorlica2005click, agarwal2009skewed, mahdian2008pay, nazerzadeh2013dynamic}.

The proposed mechanism is distinguished by two key differences from existing works that also investigate online auction design \citep{devanur2009price, nazerzadeh2013dynamic, babaioff2009characterizing}. For exact truthfulness, strict characterizations on attainable regret rates have been established for deterministic payment rules \citep{devanur2009price, babaioff2009characterizing}. Weaker asymptotic truthfulness has also been considered under the pay-per-action framework known in online advertising \citep{immorlica2005click,mahdian2008pay,nazerzadeh2013dynamic}. Most related to our work is \citep{nazerzadeh2013dynamic}, which studies the pay-per-action setting, which allows reporting value after allocation, and proposes a history-dependent pricing rule. Our work builds on this progress by studying contextual auctions, removing the history-dependent pricing rule, and eliminating the requirement for an exact value offering a human-aware perspective and a more computationally-efficient mechanism. See Table \ref{tab:comparison} for a comparison between these approaches and our proposal. 

\subsection{Partially-informed auctions}

Auction design with incomplete information has been a topic of interest in recent years \citep{bergemann2007information, feng2018learning, epasto2021clustering, guo2022no}. In particular, \citep{bergemann2007information} considers single-item multi-bidder auctions where information is only partially revealed to bidders. \citep{feng2018learning} investigates the single-item setting where bidders learn to bid with partial feedback and obtain no-regret against the best fixed bid in hindsight. \citep{guo2022no} extend this analysis by considering context and propose no-regret algorithms that are efficient from the buyer's perspective with applications to privacy. However, all of this prior work requires that the agent learn-to-bid, which requires additional effort and is commonly understood to lower an agent's welfare \citep{cai2015optimum}. Our work takes a different approach to study partially-informed auctions by focusing on user compatibility to provide their value after allocation while still maintaining the connection to these privacy considerations \citep{epasto2021clustering}. See Table \ref{tab:comparison} for a comparison between \citep{guo2022no} and our proposal. 

\subsection{Welfare and truthful elicitation}

In the realm of data pricing and acquisition, we draw from research on learning-based data pricing \citep{chen2023learning, zhao2021right, karimireddy2022mechanisms}, peer-prediction mechanisms \citep{prelec2004bayesian, witkowski2012robust, cai2015optimum}. In particular, \citep{cai2015optimum} develops a model for constructing statistical estimators in the presence of costly information revelation, while \citep{prelec2004bayesian} and \citep{witkowski2012robust} propose peer-prediction techniques for eliciting truthful information without the need for ground-truth data. These works have made significant contributions to understanding optimal mechanisms but are more concerned with obtaining truthful responses rather than the socially efficient allocation of goods. 

In recent years, the question of worker welfare during data acquisition has become a central issue. In particular, recent work on OpenAI's toxicity filter had to be halted because data annotation had a traumatic effect on workers \citep{perrigo2023exclusive}. Toxic data annotation, in general, is known to have traumatic effects on workers \citep{burns2008emotional, arsht2018human, steiger2021psychological, perrigo2023exclusive}. In particular, \citep{steiger2021psychological} proposes preventing or reducing exposure as a potential technological intervention strategy. One of our contributions is to formalize this problem as an instance of a reverse auction (negative utilities) in our setting and demonstrate theoretically and empirically that we can asymptotically maximize the social welfare of workers. 

\begin{table}[ht]
    \centering
    \smallskip
    \caption{Comparison of our mechanism with prior work}
    \smallskip
    \label{tab:comparison}
    \begin{tabular}{c c c c c c}
      \hline
      \textbf{Auction} & \textbf{Strategy}  & \textbf{Feedback} & \textbf{Local} & \textbf{Efficient} \\
       \textbf{Mechanism} & \textbf{Robust} & \textbf{Reporting} & \textbf{Payments} & \textbf{User-Strategy} \\
      \hline
      \citep{devanur2009price} & \checkmark & $\times$ & \checkmark &  $\times$ \\
      \hline
       \citep{nazerzadeh2013dynamic} & \checkmark  & \checkmark & $\times$ & $\times$ \\
      \hline
      \citep{guo2022no} & \checkmark & $\times$ & \checkmark & $\times$ \\
      \hline
      \textbf{This Work} & \checkmark & \checkmark & \checkmark & \checkmark \\
      \hline
    \end{tabular}
    \smallskip
\end{table}

\section{Preliminaries}
\label{background}

In this section, we overview the problem setting under consideration and introduce our key definitions. Discussion of how to learn a payment rule is discussed in Section \ref{model}. As an example (Figure \ref{fig:Motivation}), agents could be competing for resources to generate artwork, where the space of prompts is $\mathcal{W}$, the space of digital goods is $\mathcal{O}$, and the mechanism determines the price and allocation of resources. We summarize our notation in Appendix \ref{notationSum} provided in the supplementary materials. 

\textbf{Problem Setting:} We consider a scenario where a set of $n$ agents compete for allocations across discrete periods $t = 1, 2, \ldots, H$, up to a horizon $H$. A mechanism $\mathcal{M}$ oversees pricing for allocations. At each period $t$ the following happens:

\begin{enumerate}
    \item Each agent $i \in [n]$ submits a request $w_{it} \in \mathcal{W}$ sampled independently from one another, and a Boolean-valued array of allocations $x_{it}$ is generated for the agents.
    \item If the $j^{\text{th}}$ agent receives an allocation, then the agent receives a digital good $o_{jt} \in \mathcal{O}$ sampled some distribution conditioned on $w_{jt}$ which is sold to the agent.
    \item The non-negative value of the agent who receives an allocation during period $t$ is denoted by a bounded random variable $u_{jt}: \mathcal{W} \times \mathcal{O} \to [0,1]$.
    \item The agent pays an amount $p_{jt}$ determined by the mechanism and then reports a Boolean variable $r_{jt}(c)$ indicating if $u_{jt}$ is above some value $c$ which is to be chosen randomly. 
\end{enumerate}

We'll emphasize that $u_{it}$ has randomness from the requests $w_{it}$ and the mechanism that generates outputs $o_{it}$. Since agent's know their requests, we'll use $u_{it}(w)$ to denote the utility random variable given a request $w$. Our main assumption is an independence assumption on the requests.

\begin{assumption}\label{assumption:independent}
    For each agent $i \in [n]$, their request and output sequences $w_{it}$ and $o_{it}$ are independent of other agents and allocations.
\end{assumption}

Assumption~\ref{assumption:independent} allows sequential generation, such as written content, but precludes collusion among agents or exploiting specific prompts for high utility. We also introduce notation for an ideal setting where we have perfect knowledge of the agents' expected values. 

\begin{equation}
\label{allocation}
    \mu_{i}(w_{it}) := \mathbb{E}[u_{it}(w_{it}) | \lbrace (w_{ik}, u_{ik} )\rbrace_{k < t}], \ x_{it} := \mathbb{I}(\mu_{i}(w_{it}) \ge \mu_{j}(w_{jt}) \ \forall j)
\end{equation}

This defines the agent's expected utility and allocation. We will design $\mathcal{M}$ (Section \ref{model}) so that agents have no-regret for participating and truthful reporting along with comparable revenue and social welfare to a standard auction format. To introduce the key definitions, we just need to know $\mathcal{M}$ will determine allocations $\hat x_{it}$ and prices $p_{it}$ using value estimates $\hat \mu_{it} \sim \mu_{i}$ from the reports.

\begin{definition}
The $i^{\text{th}}$ agent is considered truthful if $r_{it}(c) = \mathbb{I}(u_{it} \ge c)$ for all $t$ and $c \in [0,1]$. 
\end{definition} 

This means they respond to the queries in Fig \ref{fig:Motivation} accurately. In our setting, there are numerous reporting strategies $r_{it} : (\mathcal{W} \times \mathcal{O}) \times [0,1] \to \lbrace 0 , 1 \rbrace$ the agents could use. Our focus lies in designing a mechanism such that agent incentives are aligned with truth-telling. Another important criterion for each agent is that they have no-regret for participation.

\begin{definition} \label{asympIR} $\mathcal{M}$ is asymptotically ex ante individually rational if each agent $i \in [n]$ has no-regret for participation when they are truthful. Specifically, the long-term total utility of the agent is nonnegative:
\[
\liminf_{H \to \infty} \mathbb{E} \left[ \sum_{t = 1}^H \hat x_{it} \mu_{i} - p_{it} \right] \ge 0
\]
\end{definition}

While definition \ref{asympIR} captures the rationality of each truthful agent's participation, it does not guarantee no-regret against other reporting strategies.

\begin{definition} Let $U_i(H)$ be the expected total utility of the $i^{\text{th}}$ agent using a truthful reporting strategy and $\hat U_i(H)$ be the maximum expected profit whenever all other agents are truthful. We say that $\mathcal{M}$ is asymptotically truthful if truthful reporting is no-regret against strategic reporting:
\[
\hat U_i(H) - U_i(H) = o(H)
\]
\end{definition}

This definition is similar to previous definitions in that it ensures that deviating from truthful reporting is relatively unprofitable \citep{pavan2009dynamic, nazerzadeh2013dynamic}. The main distinction is that it is regret-based which enables us to obtain rates in our analysis. While an asymptotic definition seems limiting, strong notions, such as dominant strategy incentive compatibility, are achievable only in limited settings \citep{pavan2009dynamic, kakade2013optimal}.

We also desire $\mathcal{M}$ to have no-regret against an idealized auction. We compare the welfare and revenue of our mechanism to a baseline given by the second-price auction known to be welfare and revenue maximizing \citep{myerson1981optimal}. In this format, allocations go to the highest bidder, say, the $i^{\text{th}}$ agent who will pay $\gamma_t = \max_{j \not = i} \mu_{j}(w_{jt})$ to $\mathcal{M}$. Otherwise, they pay nothing.

\begin{definition} We say $\mathcal{M}$ is asymptotically ex-ante welfare maximizing if it has no-regret against the welfare generated by a second-price auction:
\label{welfare_def}
\[
\mathbb{E} \left[ \sum_{t = 1}^H \sum_{i = 1}^n \hat x_{it} \mu_{i}(w_{it}) \right] - \mathbb{E} \left[ \sum_{t = 1}^H \max_i(\mu_{i}(w_{it})) \right] = o(H)
\]
\end{definition}

\begin{definition} We say $\mathcal{M}$ is asymptotically equivalent to the revenue of the second-price auction if it has no-regret against the revenue generated by a second-price auction:
\[
\mathbb{E} \left[ \sum_{t = 1}^H \sum_{i = 1}^n \hat x_{it} p_{it} \right] - \mathbb{E} \left[\sum_{t = 1}^H \gamma_t \right] = o(H)
\]
\end{definition}

\section{The proposed mechanism}
\label{model}

\begin{table}[ht]
\vspace{0pt}
\centering
\begin{algorithm}[H]
        \caption{Feedback-Driven Mechanism}
        \label{alg:Mechanism}
        \SetKwInOut{Input}{Input}
        \SetKwInOut{Output}{Output}
        \Input{Exploration rate $\eta_t$, agent submissions $w_{it}$}
        \Output{Tuple of context-report pairs}
        \For{$t = 1, 2, \cdots$}
        {
          \eIf{$\text{explore with probability } \eta_t$}
          {
            $i = \text{sample}([1, ... , n])$ \;
            $x_{it} = 1$ \;
            $c_{it} = \text{sample}([0,1])$ \;
            $p_{it} = 0$ \;
            $r_{it}(c_{it}) = \text{agent-report}(w_{it}, c_{it})$ \;
          }
          {
            $i = \text{argmax}_{j} \hat \mu_{jt}(w_{jt})$ \;
            $x_{it}, y_{it} = 1$ \;
            $p_{it}, c_{it} = \max_{j \not = i} \hat \mu_{jt}(w_{jt})$ \;
            $r_{it}(c_{it}) = \text{agent-report}(w_{it}, c_{it})$ \;
          }
        }
\end{algorithm}
\end{table}

In our approach, as illustrated in Figure \ref{fig:Motivation} and implemented in Algorithm \ref{alg:Mechanism}, we aim to estimate the utility function of each agent using a learning algorithm $\mathcal{L}$, connecting with the high-level goals of the paper by designing an auction mechanism that improves welfare in digital goods and data auctions. Ideally, we would know $\mu_{i}$ for each agent $i \in [n]$ and allocate using a second-price auction. The challenge lies in the potential misreporting of observed utilities by agents to gain utility.

The basic mechanism is a second-price payment rule estimated with a learning algorithm $\mathcal{L}$. We fit $\hat \mu_{it}$ using $\mathcal{L}$ and a data set of context and reporting tuples $\lbrace (w_{ik}, r_{ik}, c_{ik} )\rbrace_{k \in S_{it}}$, where $c_{ik} = p_{ik}$ represents the price comparison, and $S_{it}$ denotes periods of allocation to the $i^{\text{th}}$ agent up to time $t$. Simultaneously, the proposed mechanism performs a variant of $\epsilon$-greedy allocation, allocating to agents with the highest estimated value during exploitation rounds indicated by $y_{it}$ or exploring by allocating for free to a randomly chosen agent with probability $\eta_t \in [0,1]$. During exploration rounds, agents still compare to a price point $c_{it}$ sampled from a distribution over $[0,1]$. Finally, the exploitation round payments are $\hat \gamma_t = \text{max}_{j \not = i} \hat \mu_{jt}(w_{jt})$, with $i$ indicating the allocated agent. In general, the payment $p_{it}$ equals $y_{it} \hat \gamma_t$.

To study allocation and payment, we define the best empirical estimate under $\mathcal{L}$ of an agent's expected utility and allocation using the data set $\lbrace (w_{ik}, c_{ik}, r_{ik})\rbrace_{k \in S_{it}}$. 

\begin{equation}
    \hat \mu_{it}(w_{it}) := \mathbb{E}_{\mathcal{L}}[u_{it} | \lbrace (w_{ik}, c_{ik}, r_{ik})\rbrace_{k \in S_{it}}], \ \hat x_{it} := \mathbb{I}(\hat \mu_{it}(w_{it}) > \hat \mu_{jt}(w_{jt}) \ \forall j)
\end{equation}

where $\mathbb{E}_{\mathcal{L}}$ is an estimate under $\mathcal{L}$ for the true expected value. Discussion of a concrete choice of a learning algorithm is delayed until Section \ref{pair}. It is worth remarking that these definitions differ from equation \ref{allocation} because the agent merely reports their relative utility against $c_{ik}$. 

It is worth making a few remarks comparing our mechanism (see Table \ref{tab:comparison}) to related works implementing online auctions with learned expected values \citep{devanur2009price, nazerzadeh2013dynamic, babaioff2015truthful, guo2022no}. Our design deviates significantly in key areas. We allow agents to self-report their satisfaction, circumventing the "learn-to-bid" assumption adopted by \citep{guo2022no}. We introduce a simplified reporting rule that directly links reports to current payments via binary feedback, which simplifies value reporting for agents, a notable contrast from all of these works. This strategy also contrasts with the computationally demanding history-dependent payment rule used by \citep{nazerzadeh2013dynamic}, which scales quadratically. This tailored approach marks our mechanism as both efficient and user-side oriented, as per our prior technical discussion.

\section{Analysis}
\label{analysis}

We develop sufficient conditions for the learning algorithm $\mathcal{L}$ applied in Algorithm \ref{alg:Mechanism} to estimate the $\mu_i$ that results in a mechanism $\mathcal{M}$ that meets our social efficiency criterion. We then explore how to implement the learning algorithm and examine a relevant data acquisition example. Our main condition involves the error from $\mathcal{L}$ applied to truthful reporting data from exploration rounds:

\begin{equation}
\Delta_t := \mathbb{E}[ \max_k |\mu_{k}(w_{kt}) - \hat \mu_{kt}(w_{kt} | x_{kt'} = 1, y_{kt'} = 0, 0 \le t' < t)|]
\end{equation}

As the number of exploitation allocations is dependent on the other agents' behavior, we focus on performance using just the randomly allocated exploration rounds. Our main result offers valuable insights into the performance of mechanism $\mathcal{M}$ using intuitive conditions on the learning algorithm $\mathcal{L}$.

\begin{restatable}[]{theorem}{mainresult}
\label{main_result}
Suppose our mechanism $\mathcal{M}$ estimates agent values bounded to the unit-interval using some learning algorithm $\mathcal{L}$. Suppose the expected error of this algorithm is monotone decreasing in the number of samples and that for all time,
\begin{equation}
\sum_{t = 1}^H \Delta_t = o \left( \sum_{t = 1}^H \mathbb{E}[\eta(t)] \right) = o(H)
\end{equation}
then the mechanism  $\mathcal{M}$ satisfies the following:

\begin{enumerate}
    \item Is asymptotically individually rational and asymptotically truthful
    \item Is asymptotically welfare and revenue maximizing. 
    \item Compared to a second-price auction $\mathcal{M}$ can obtain welfare and revenue regret $\tilde O(H^{2/3})$ if there is a learning algorithm $\mathcal{L}$ for valuations with slow learning rate $\tilde O(H^{-1/2})$ and regret $\tilde O(H^{1/2})$ if there is an algorithm with a fast rate $\tilde O(H^{-1})$.
\end{enumerate}
\end{restatable}

The primary assumption is that the expected sum of errors for the learning algorithm decreases quickly relative to the sum of expected regret terms $\mathbb{E}[\eta(t)]$ over the entire horizon $H$. Although we can always increase the number of exploration rounds for a consistent algorithm to meet this condition, doing so may significantly reduce revenue. We discuss a concrete learning algorithm in Section \ref{pair}. 

The proof of this result relies on determining whether strategic reporting can be profitable for an agent, which we bound in terms of the error induced by the learning algorithm. We outline the main steps here and defer the proof to Appendix \ref{Main_Proof} in the supplementary materials. 

\begin{restatable}[]{lemma}{strategicutility}
\label{strategicutility}
A single agent providing misleading feedback can increase their expected utility up to time $H$ by no more than $6 \Sigma_{t \in [H]} \Delta_t$. 
\end{restatable}

\begin{hproof}
    The proof of Lemma \ref{strategicutility} relies on evaluating the expected utility of an agent deviating from a truthful strategy. We first ignore exploration allocations since they are strategy independent. We then consider an agent who deviates from a truthful reporting strategy $T$ to another strategy $L$. We fix $\hat \mu_{it}$ as the utility estimated from data collected under strategy $T$. We then analyze the expected utility of both strategies. The expected profit of such a strategy is given by the difference between the expected utilities under the new strategy $L$ and the truthful strategy $T$. We decompose this profit into three cases, corresponding to the different allocation times for the item to the agent: those in $S_L \setminus S_T$, those in $S_T \setminus S_L$, and those in $S_T \cap S_L$. For each case, we establish bounds on the differences between the expected utilities of the two strategies. These bounds involve the estimates of utilities for different agents under both strategies and the true utilities of the allocated items. We then combine the results of these cases and obtain an upper bound on the profit of deviating from the truthful strategy, which is $O( \mathbb{E}[\Sigma_t \Delta_t] )$ over all time steps up to time $H$.
\end{hproof}

This bound is worse than \citep{nazerzadeh2013dynamic} by a constant factor since we do not use a history-dependent payment rule which can correct for its past estimation errors. Unlike this work, we use a regret-based framework to proceed with the rest of the analysis. 

\begin{restatable}[]{lemma}{asympIC}
\label{asympIC}
    $\mathcal{M}$ is asymptotically truthful and individually rational. 
\end{restatable}

\begin{hproof}
    Notice that the expected profit from exploration rounds is the same between strategies. Therefore, by Lemma \ref{strategicutility} we have $\hat U_i(H) - \text{U}_i(H) \le 6 \Sigma_{t \in [H]} \Delta_t$. We also have asymptotic individual rationality because we can bound the overcharges to the agent,
\[
\mathbb{E} \left[ \sum_{t \in S_T} \hat \gamma_t - \mu_{i}(w_{it}) \right] \le \mathbb{E} \left[ \sum_{t \in S_T} \hat \mu_{it}(w_{it}) - \mu_t(w_{it}) \right] \le \mathbb{E} \left[ \sum_{t} \Delta_t \right]
\]
As we send the horizon to infinity, we see that these overcharges are bounded by the error rate. By assumption, this term is dominated by the free allocations, so we have asymptotic ex-ante individual rationality. 
\end{hproof}

Now we examine the question of regret concerning welfare and revenue objectives. Our mechanism may lose revenue during exploration and due to estimation error during exploitation.

\begin{restatable}[]{lemma}{asympW}
\label{asympW}
    $\mathcal{M}$ is asymptotically welfare and revenue maximizing. Compared to a second-price auction $\mathcal{M}$ can obtain welfare and revenue regret $\tilde O(H^{2/3})$ if there is a learning algorithm $\mathcal{L}$ for valuations with slow learning rate $\tilde O(H^{-1/2})$ and regret $\tilde O(H^{1/2})$ if there is an algorithm with a fast rate $\tilde O(H^{-1})$.
\end{restatable}

\textbf{Remarks:} Our results improve upon \citep{nazerzadeh2013dynamic}, which does not provide finite-time regret rates for their algorithm. Furthermore, compared to \citep{guo2022no}, which only considers a single agent against a post price, our mechanism considers a multi-agent setting without assuming the agents' ability to bid. In this work, they establish this bound for applications to user privacy where the context is masked in order to preserve privacy. Since masking is deterministic in this work, we can use standard results from realizable learning theory, see \citet{anthony1999neural}, to conclude our algorithm also achieves $\tilde O(H^{-1/2})$ in the stochastic setting. In general, it is unclear if using an adaptive algorithm would help given $\Omega(H^{2/3})$ lower-bounds on regret in this setting \citep{devanur2009price, babaioff2009characterizing}. Despite these limitations, our mechanism showcases the importance of feedback and welfare in shaping digital goods and data auction mechanisms that are more efficient and user-friendly. 

\subsection{Estimating the value model with pairwise feedback}
\label{pair}

To make our theory concrete, we can consider $\mu_{i}$ that follows a linear model with $d$-features and with standard independent noise and realizability assumption. We will call the resulting mechanism using uniformly sampled comparison prices, linear regression algorithm $\mathcal{L}$, and exploration rate $\eta$ by $\mathcal{M}_{\eta}(\text{linear})$.

\begin{restatable}[]{theorem}{pairMech}
\label{pairMech}
For an exploration rate $\eta_t = t^{-1/3} \cdot (n \log(t))^{(1+2 \epsilon)/3}$  we have that $\mathcal{M}_{\eta}(\text{linear})$ satisfies the conditions of Theorem \ref{main_result} and so is asymptotically individually rational, incentive compatible, and no-regret in revenue and welfare. 
\end{restatable}

The full proof is provided in Appendix \ref{Example_Proof} in the supplementary materials. The main technical step in the proof is to identify the underlying utility from pairwise reports. In our mechanism, we obtain reports $r_{it}$ from agents regarding if the random utility $u_{it}$ derived from their allocation for a request $w_{it}$ is satisfactory. This is based on being above or below some reference payment $c$. Our model is that $r_{it}(c) = \mathbb{I}(u_{it} \ge c)$. In particular, we have the following,

\begin{restatable}[]{lemma}{meanPair}
    \label{meanPair}
    If $u$ is a nonnegative random variable, we have that,
    \begin{equation}
    \mu = \int_0^{1} \mathbb{P}(u \ge c) d c
    \end{equation}
    interpreting the integral as a Lebesgue integral with respect to the Lebesgue measure.
\end{restatable}

An immediate corollary of Lemma \ref{meanPair} is that the least-squares estimator for the conditional expectation $\mathbb{E}[r_i | w]$ given by $\hat r_i(w)$ equals $\hat \mu_i(w)$ in the population setting under uniform random sampling of comparison prices. Another implication of this result is that we also obtain a data set of labeled comparisons expressing human feedback on the performance of the underlying generative process. One potential application is in the context of reinforcement learning from human feedback \citep{christiano2017deep}. In this setting, we would also allow comparisons between generated outputs and then train the value model using these pairwise comparisons as constraints. 

\subsection{Experiment with toxicity annotation}
\label{toxicity_setting}

\begin{figure*}[!t]
    \centering
    \begin{subfigure}{0.45\linewidth}
        \includegraphics[width=\linewidth]{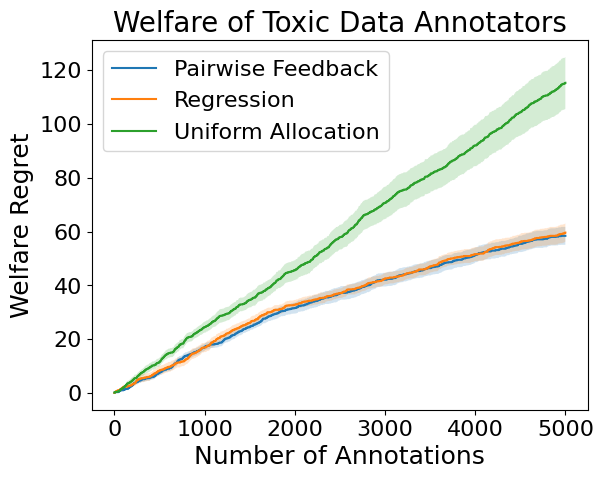}
        \caption{Welfare regret}
        \label{fig:Welfare_Regret}
    \end{subfigure}
    \hfill
    \begin{subfigure}{0.45\linewidth}
        \includegraphics[width=\linewidth]{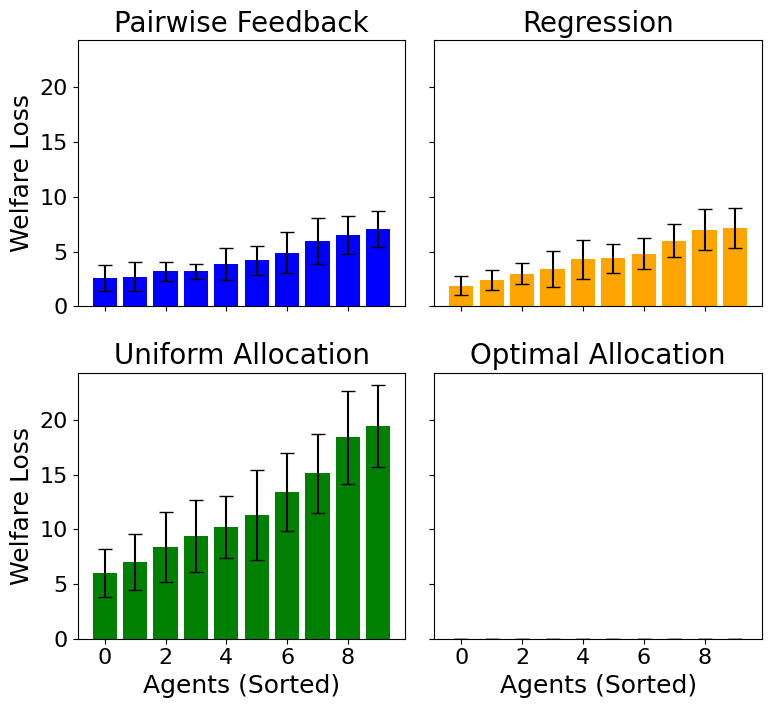}
        \caption{Welfare histogram}
        \label{fig:Welfare_Histogram}
    \end{subfigure}
    \caption{Welfare is the sum of the utilities of all agents across all the allocation periods. We compare the expected welfare regret (a) of different allocation mechanisms for data acquisition vs. an ideal allocation mechanism. Our relative feedback approach elicits relative utility information, regression fits values using utility reports from agents, and uniform allocation methods such as peer prediction make assignments at random. The welfare histogram (b) shows the distribution of welfare losses for each agent under different allocation mechanisms.}
    \label{fig:Welfare}
\end{figure*}

In some cases, allocations to users might result in negative utility, which would mean the mechanism pays users for reporting feedback. For instance, the employment of low-wage workers to enhance AI systems has brought about ethical concerns, such as the distressing impact on workers who review harmful content \citep{steiger2021psychological,perrigo2023exclusive}. In particular, \citep{steiger2021psychological} proposes preventing or reducing exposure as a potential technological intervention strategy. In this section, we present two experiments to evaluate our mechanism as an intervention strategy \footnote{We provide the code for our experiments in the supplementary materials}. 

For the experiments, we use multi-label data from the Toxic Comment Classification Challenge on Kaggle, which contains a large number of Wikipedia comments labeled by human raters for toxic behavior, including categories such as toxic, severe-toxic, obscene, threat, insult, and identity-hate \citep{Kaggle_2018}. We assess the welfare regret of various allocation strategies by comparing them to the strategy that allocates resources to the agent with the highest expected welfare during each period. This concept is formally defined in definition \ref{welfare_def}. The expected utility is modeled with a linear model, as discussed in Section \ref{pair}, and employs 30-dimensional PCA analysis to GloVe features for the data representation \citep{pennington2014glove}. Our experiment involves 10 agents and spans 5000 rounds of allocation. We assume that each agent possesses a fixed type sensitivity of which they are unaware, which determines their utility function. Every agent classifies examples sampled i.i.d from the dataset as toxic or non-toxic by reporting their relative utility from viewing the example. Moreover, when an agent encounters an example with a label matching their type sensitivity, they lose one unit of utility. For instance, an agent may be particularly sensitive to obscene examples.

We evaluate the welfare performance of three mechanisms: our method based on relative feedback, approaches that directly regress utility, and uniform assignment approaches. The direct utility regression methods, including \citep{devanur2009price, babaioff2009characterizing, nazerzadeh2013dynamic}, vary in terms of payment structures, while learn-to-bid methods \citep{guo2022no} assume workers estimate their own value and learn to bid accordingly. Uniform assignment approaches, on the other hand, make no attempt to intervene in the content allocation process and encompass most peer-prediction methods that pay based on conformity rather than utility, focusing on incentive compatibility issues \citep{prelec2004bayesian, witkowski2012robust, cai2015optimum}.

Our experimental results in Figure \ref{fig:Welfare_Regret} and \ref{fig:Welfare_Histogram} show that using an auction mechanism can significantly improve the welfare of the allocations given to agents over the peer-prediction method and performs favorably to the optimal allocation strategy. Moreover, our relative elicitation mechanism is competitive with the full information set but has the advantage of being simpler for users to report on. For example, while we make no further modifications to the calculation of welfare beyond what has already been discussed, some works assume there is a further cost to complicated elicitation strategies \citep{cai2015optimum}. 

\section{Limitations and future work}

\label{sec:limitations}

Our proposed approach has some limitations that warrant further exploration. While we propose an intervention method to improve the social welfare of workers doing toxic annotation, other aspects, such as negative psychological impacts and systemic issues around who does such work, are left unaddressed. While we provide a mechanism for the dynamic setting, we assume value evolves independently of other agents and the mechanism allocations, which prohibits us from studying collusion or adversarial scenarios. In particular, agents who make alias accounts could game the mechanism for free exploration allocations. Also, it is possible a randomized approach could improve upon the rates presented. Finally, we think further exploration as auction mechanism for selling digital goods in real-world settings is an important direction. 

\section{Conclusions and societal impact}

In this paper, we have presented a novel approach to auctioning AI services that emphasizes user-friendly bidding, extends to multi-agent and contextual settings, offers simpler mechanisms, and improved bounds. Our approach has the potential for significant societal impact by facilitating a more efficient allocation of AI services and enabling a wider range of users to access these services without requiring them to have a deep understanding of their value. At the same time, we recognize that the collection of feedback or toxic annotation data may have negative externalities on the privacy and welfare of workers. Addressing the limitations and ethical concerns identified, we can move towards a future where AI services are more accessible, efficient, and ethically responsible, ultimately leading to a positive impact on society.

\begin{ack}
I'd like to thank Ellen Vitercik for advising during the initial development of theory for the mechanism. I'd also like to thank Neil Band for suggesting I pursue simplifying the reporting method. 
\end{ack}

\bibliographystyle{unsrtnat}
\bibliography{references}

\newpage

\appendix

\section{Summary of notation}
\label{notationSum}

\begin{table}[ht]
\centering
\smallskip
\caption{Summary of notation}
\smallskip
\begin{tabular}{ll} 
    Notation & Definition \\
    \hline
    $i \in [n]$ & Agent $i$ out of $n$ \\
    $t \in [H]$ & Periods $t$ in horizon $H$ \\
    $w_{it} \in \mathcal{W}$ & Agent request \\
    $o_{it} \in \mathcal{O}$ & Produced good from request \\
    $u_{it} \in [0,1]$ & Agent $i$ utility at time $t$ \\
    $\mu_{i}$ &  Expected value function of agent $i$ \\
    $\mathcal{L}$ & Learning algorithm for $\mu_{i}$ \\
    $\Delta_t$ & $\mathcal{L}$'s error - exploration rounds \\
    $\eta_t$ & Exploration frequency \\
    $\hat \mu_{it}$ & Estimate for $\mu_{i}$ with $\mathcal{L}$\\
    $\hat \gamma_t$ & Estimated second-price \\
    $S_{it}$ & Allocations of $i$ up to time $t$ \\
    $x_{it} \in \lbrace 0, 1 \rbrace$ & Allocation indicator \\
    $\hat x_{it}$ & Allocations with $\mathcal{L}$ \\
    $y_{it} \in \lbrace 0, 1 \rbrace$ & Exploitation indicator \\
    $r_{it} \in \lbrace 0, 1 \rbrace$ & Report of agent $i$ at time $t$  \\
    $p_{it}$ & Payment of agent $i$ at time $t$ \\
    $c_{it}$ & Comparison price \\ 
    \hline
\end{tabular}
\smallskip
\end{table}

\section{Omitted proofs}

\subsection{Proof of theorem \protect\ref{main_result}}
\label{Main_Proof}

\mainresult*

We proceed by first showing that strategic reporting isn't particularly profitable compared to the profit from free allocations. Following this, we establish asymptotic truthfulness and individual rationality properties. Finally, we analyze the welfare and revenue regret of our mechanism. 

\strategicutility*

\begin{proof}

Recall that $\Delta_t$ is the expected maximum error of the learning algorithm over all the agents $i \in [n]$ with respect to exploration round data. We denote the set of allocation times with $S := \lbrace t \in [H] : \hat x_{it} = 1 \rbrace$. Recall there are numerous reporting strategies $r_{it} : (\mathcal{W} \times \mathcal{O}) \times [0,1] \to \lbrace 0 , 1 \rbrace$ each agent could use. Accordingly, when it is important to specify the reporting strategies used by the agents, we will adorn these definitions in some fashion with the relevant strategy $L$ e.g., $S_L$, $\hat x_{it}(L), \hat \mu_{it}(w_{it} | L)$.

We ignore exploration allocations since they are strategy independent and therefore cancel each other out. Still, our agent may try to deviate from being truthful under some reporting strategy $T$ to some other reporting strategy $L$ that induces a new set of allocation times $S_L$ for the item to the agent. We fix $\hat \mu_{it}$ to represent the utility estimated from data collected under strategy $T$. By the tower property the expected utility for a given strategy satisfies,
\[
U_i(L, H) = \mathbb{E} \left[\sum_{t \in S_L} u_{it}(w_{it}) - \hat \gamma_t(L) \right] 
\]
\[
= \mathbb{E} \left[ \sum_{t \in S_L} \mathbb{E} \left[ u_{it}(w_{it}) - \hat \gamma_t(L) \big \vert \lbrace (w_{ik}, c_{ik}, r_{ik})\rbrace_{k \in [t]} \right] \right]
\]
\[
= \mathbb{E} \left[ \sum_{t \in S_L} \mu_{i}(w_{it} ) - \hat \gamma_t(L) \right]
\]

Accordingly, the expected profit of such a strategy is given by,

\[
\text{Profit}(L,H) = U_i(L,H) - U_i(T,H)
\]
\[
= \mathbb{E} \left[ \sum_{t \in S_L} \mu_{i}(w_{it}) - \hat \gamma_t(L) \right] - \mathbb{E} \left[ \sum_{t \in S_T} \mu_{i}(w_{it}) - \hat \gamma_t(T) \right]
\]
\[
= \mathbb{E} \left[\sum_{t \in S_L \setminus S_T} \mu_{i}(w_{it}) - \hat \gamma_t(L) \right] - \mathbb{E} \left[\sum_{t \in S_T \setminus S_L} \mu_{i}(w_{it} ) - \hat \gamma_t(T) \right]  + \mathbb{E} \left[\sum_{t \in S_T \cap S_L} \hat \gamma_t(T) - \hat \gamma_t(L) \right]
\]

In the last line, observe that whenever $t \in S_T \setminus S_L$ we have that $\hat \mu_{it}(w_{it} | T) > \hat \gamma_t(T)$ and so,
\[
\mu_{i}(w_{it} ) - \hat \gamma_t(T) = \mu_{i}(w_{it} ) - \hat \mu_{it}(w_{it} | T) + \underbrace{\hat \mu_{it}(w_{it} | T) - \hat \gamma_t(T)}_{> 0}\]
\[\Rightarrow \mu_{i}(w_{it} ) - \hat \gamma_t(T) \ge \mu_{i}(w_{it}) - \hat \mu_{it}(w_{it} | T), \quad t \in S_T \setminus S_L
\]
For $t \in S_L \setminus S_T$ let $j$ be the agent who would receive the item if agent $i$ were truthful. We have $\hat \mu_{jt}(w_{jt} | T) \ge \hat \mu_{it}(w_{it} | T)$. So we have,
\[
\mu_{i}(w_{it}) - \hat \gamma_t(L) \le \mu_{i}(w_{it}) - \hat \mu_{it}(w_{it} | T) + \hat \mu_{jt}(w_{jt} | T) - \hat \gamma_t(L) \quad t \in S_L \setminus S_T
\]
We also have that $\hat \mu_{jt}(w_{jt} | L) \le \hat \gamma_t(L)$ so we have,
\[
\mu_{i}(w_{it}) - \hat \gamma_t(L) \le (\mu_{i}(w_{it}) - \hat \mu_{it}(w_{it} | T)) + \max_{k \not = i} \lbrace \hat \mu_{kt}(w_{kt} | T) - \hat \mu_{kt}(w_{kt} | L) \rbrace \quad t \in S_L \setminus S_T
\]
Finally, for $t \in S_T \cap S_L$ we let $j \not = i$ be the agent with the highest $\hat \mu_{jt}(w_{jt} | T)$. So we have,
\[
\hat \gamma_t(T) - \hat \gamma_t(L) \le \hat \mu_{jt}(w_{jt} | T) - \hat \mu_{jt}(w_{jt} | L) \le \max_{k \not = i} \lbrace \hat \mu_{kt}(w_{kt} | T) - \hat \mu_{kt}(w_{kt} | L) \rbrace
\]
After substitution, we end up with a simplified expression for the $i^{\text{th}}$ agent's profit:

\begin{equation}
    \label{final_profit}
\text{Profit}(L,H) \le \mathbb{E} \left[\sum_{t \in S_L \setminus S_T} (\mu_{i}(w_{it}) - \hat \mu_{it}(w_{it} | T)) + \max_{k \not = i} \lbrace \hat \mu_{kt}(w_{kt} | T) - \hat \mu_{kt}(w_{kt} | L) \rbrace \right] 
\end{equation}
\[- \mathbb{E} \left[\sum_{t \in S_T \setminus S_L} \mu_{i}(w_{it} ) - \hat \mu_{it}(w_{it} | T) \right]  + \mathbb{E} \left[\sum_{t \in S_T \cap S_L} \max_{k \not = i} \lbrace \hat \mu_{kt}(w_{kt} | T) - \hat \mu_{kt}(w_{kt} | L) \rbrace \right] \]

To proceed, we observe that:
\[
\max_{k \not = i} \lbrace \hat \mu_{kt}(w_{k} | T) - \hat \mu_{kt}(w_{kt} | L) \rbrace \]

\[\le \max_{k \not = i} |\mu_{k}(w_{k}) - \hat \mu_{kt}(w_{kt} | T)| + |\mu_{k}(w_{kt}) - \hat \mu_{kt}(w_{kt} | L)|
\]
\[
\le \max_{k \not = i} |\mu_{k}(w_{kt}) - \hat \mu_{kt}(w_{kt} | T)| + \max_{k \not = i}  |\mu_{k}(w_{kt}) - \hat \mu_{kt}(w_{kt} | L)|\]

Recall the assumption that only the agent $i$ is misreporting. Therefore, all the agents except $i$ are truthful so each term can be bounded by $\Delta_t$ and so we end up with,

\[
\max_{k \not = i} \lbrace \hat \mu_{kt}(w_{kt} | T) - \hat \mu_{kt}(w_{kt} | L) \rbrace \le 2 \max_k | \mu_{k}(w_{kt}) - \hat \mu_{kt}(w_{kt} | T)|
\]
\[
\Rightarrow \mathbb{E}[\max_{k \not = i} \lbrace \hat \mu_{kt}(w_{kt} | T) - \hat \mu_{kt}(w_{kt} | L) \rbrace] \le 2 \Delta_t
\]

Here we are using the assumption that the learning algorithm is proper since $\Delta_t$ is the expected error over \textit{only} the exploration rounds. Substituting the result into Eq. \ref{final_profit} we obtain:

\[
\text{Profit}(L,H) \le 6 \sum_{t = 1}^H \Delta_t 
\]

as claimed, which completes the proof. 

\end{proof}

\asympIC*

\begin{proof}
    Notice that we have,

\[
\text{U}_i(H) = \mathbb{E} \left[\sum_{t \in S_T } (\mu_{i}(w_{it})-\hat \gamma_t) \right] + \frac{1}{n} \mathbb{E} \left[ \sum_{t = 1}^H \eta_t \mu_{i}(w_{it}) \right]
\]

\[
\hat U_i(H) - \text{U}_i(H) \le 6 \Sigma_{t \in [H]} \Delta_t  = o \left( \mathbb{E} \left[ \sum_{t} \eta_t \right] \right) = o(H)
\]

The last step follows from Lemma \ref{strategicutility} and since, by assumption, the free allocation rate strictly dominates the error rate of the estimation procedure.

We also have asymptotic individual rationality because we can bound the overcharges to the agent,
\[
\mathbb{E} \left[ \sum_{t \in S_T} \hat \gamma_t - \mu_{i}(w_{it}) \right] \le \mathbb{E} \left[ \sum_{t \in S_T} \hat \mu_{it}(w_{it}) - \mu_i(w_{it}) \right] \le \mathbb{E} \left[ \sum_{t} \Delta_t \right]
\]
As we send the horizon to infinity, we see that these overcharges are bounded by the error rate and, by assumption, this term is dominated by the free allocations, so we have,

\[\mathbb{E} \left[ \sum_{t \in S_T} \hat \gamma_t - \mu_{i}(w_{it}) \right] \le \mathbb{E} \left[ \sum_{t} \Delta_t \right] = o \left( \mathbb{E} \left[ \sum_{t} \eta_t \right] \right) = o(H)\]

which establishes asymptotic ex-ante individual rationality. 
\end{proof}

\asympW*

\begin{proof}
We present the argument for revenue regret here. The argument for welfare regret goes similarly. The revenue minus free allocations satisfies,
\[
\text{Revenue}(H) = \mathbb{E} \left[ \sum_{t = 1}^H (1 - \eta_t) \hat \gamma_t \right] \ge \mathbb{E} \left[ \sum_{t = 1}^H (1 - \eta_t) (\gamma_t - \Delta_t) \right]
\]

\[
\ge \mathbb{E} \left[ \sum_{t = 1}^H \gamma_t \right] - \mathbb{E} \left[ \sum_{t = 1}^H \eta_t \right] - \mathbb{E} \left[ \sum_{t = 1}^H \Delta_t \right]
\]

We conclude that,
\[
\text{Regret}(H) \le  \mathbb{E} \left[ \sum_{t = 1}^H \eta_t \right] + \mathbb{E} \left[ \sum_{t = 1}^H \Delta_t \right]
\]
We call $N_t = \sum_{k = 1}^t \eta_k$ the expected number of exploration rounds and $N_t(i)$ the number of exploration rounds given to the $i^{\text{th}}$ agent. Now we consider events when they are near their expectation. We use the following multiplicative Chernoff bound for sums of random variables bounded to the unit-interval \footnote{\url{https://en.wikipedia.org/wiki/Chernoff_bound}}:

\[ \mathbb{P}(N_t \le (1 - \delta) \mu) \le e^{-\delta^2 \mu / 2} \]

We define the event $\mathcal{A} = \lbrace N_t \ge \frac{1}{2} \mathbb{E}[N_t] \rbrace$ to capture the situation where the empirically observed number of exploration rounds is not too far below the expectation. Using the Chernoff bound we see,
\[
\mathbb{P}(\lnot \mathcal{A}) \le e^{-\mathbb{E}[N_H] / 8}
\]
We also the event $\mathcal{B} = \lbrace N_{t}(i) \ge \frac{1}{2 n} N_t \rbrace$ that the number of exploration rounds for a particular agent $i$ is not too far from the expectation. Using the Chernoff bound once again we have,
\[
\mathbb{P}(\lnot \mathcal{B}) \le e^{- N_H / 8 n} \\
\Rightarrow \mathbb{P}(\lnot \mathcal{B} | \mathcal{A}) \le e^{-\mathbb{E}[N_H] / 16 n}
\]
Conditioned on both of these events we have,
\begin{equation}
\label{learning_bound}
\mathbb{E}[\Delta_t] \le \mathbb{E}[\Delta_t | \mathcal{A} \land \mathcal{B}] \cdot \mathbb{P}(\mathcal{A} \land \mathcal{B}) + 1 \cdot \mathbb{P}(\mathcal{A}) \cdot \mathbb{P}(\lnot \mathcal{B} | \mathcal{A}) + 1 \cdot \mathbb{P}(\lnot \mathcal{A})
\end{equation}

We will use the fact that conditional on $\mathcal{A} \land \mathcal{B}$ we have $\mathbb{E}[N_t(i)] \ge \frac{t \cdot \eta_t}{4 n}$. There will be two cases to consider. The first corresponds to a slow learning algorithm, and the second to a fast learning algorithm. 

\textbf{Case 1:} We have $\mathbb{E}[\Delta_t | N_t(i) = t] = O(\sqrt{\log(t) / t})$.

Without loss of generality, we can scale the bound and shift $t$ by a constant so that we can assume $\mathbb{E}[\Delta_t | N_t(i) = t] \le \sqrt{\log(t) / t}$. From this, we deduce that:
\[
\mathbb{E}[\Delta_t | \mathcal{A} \land \mathcal{B}]  \le \sqrt{\frac{4 n \log \left( \frac{t \eta_t}{ 4n} \right)}{t \eta_t}} \le \sqrt{\frac{4 n \log(t /4)}{t \eta_t}}\]
\[
\Rightarrow \mathbb{E}[\Delta_t] \le \sqrt{\frac{4 n \log(t /4)}{t \eta_t}} +e^{-t \eta_t / 8} +  e^{-t \eta_t / 16n}
\]
where in the last step we substitute into Eq. \ref{learning_bound}. For $t \ge 4 e$ we have,
\[
\sqrt{\frac{4 n \log(t /4)}{t \eta_t}} \ge \sqrt{\frac{4 n }{t \eta_t}}
\]
Therefore,
\[
\sup_{t \ge 4 e} \frac{e^{-t \eta_t / 16 n}}{\sqrt{\frac{4 n \log(t /4)}{t \eta_t}} } \le \sup_{t \ge 4 e} \frac{e^{-t \eta_t / 16 n}}{\sqrt{\frac{4 n}{t \eta_t}} } = \sqrt{\frac{2}{e}}\]
\[
\sup_{t \ge 4 e} \frac{e^{-t \eta_t / 8}}{\sqrt{\frac{4 n \log(t /4)}{t \eta_t}} } \le \sup_{t \ge 4 e} \frac{e^{-t \eta_t / 8}}{\sqrt{\frac{4 n}{t \eta_t}} } = \sqrt{\frac{1}{n \cdot e}}
\]
Subsequently, we have for $t \ge 4 e$ that,
\[
\mathbb{E}[\Delta_t] \le c \cdot \sqrt{\frac{n \log(t)}{t \eta_t}} \ , \quad c = 2 \cdot \left(1 + \sqrt{2/e} + \sqrt{1/e} \right)
\]
Now we take $\eta_t = t^{-1/3} \cdot (n \log(t))^{(1+2 \epsilon)/3}$ and then for $t \ge 4 e$ we arrive at,
\[
\mathbb{E}[\eta_t] = t^{-1/3} \cdot (n \log(t))^{(1+2\epsilon)/3}\]
\[
\mathbb{E}[\Delta_t] \le  c \cdot \sqrt{\frac{n \log(t)}{t^{2/3} \cdot (n \log(t))^{(1+2\epsilon)/3}}} = c \cdot t^{-1/3} (n \log(t))^{(1-\epsilon)/3}
\]
We see from an integration that $\text{Regret}(H) = O(H^{2/3} \cdot \log(H)^{(1 + 2 \epsilon) / 3})$ which establishes the result. Since $\Sigma_t \mathbb{E}[\Delta_t] = o(\Sigma_t \mathbb{E}[\eta_t])$ this regret bound is consistent with the assumption in Theorem \ref{main_result}. 

\textbf{Case 2:} We have $\mathbb{E}[\Delta_t | N_t(i) = t] = O(\log(t) / t)$.

Without loss of generality, we can scale the bound and shift $t$ by a constant so that we can assume $\mathbb{E}[\Delta_t | N_t(i) = t] \le \log(t) / t$. From this we deduce that,
\[
\mathbb{E}[\Delta_t | \mathcal{A} \land \mathcal{B}]  \le \frac{4 n \log \left( \frac{t \eta_t}{ 4n} \right)}{t \eta_t} \le \frac{4 n \log(t /4)}{t \eta_t} \]
\[\Rightarrow \mathbb{E}[\Delta_t] \le \frac{4 n \log(t /4)}{t \eta_t} +e^{-t \eta_t / 8} +  e^{-t \eta_t / 16n}
\]
where in the last step we substitute into Eq. \ref{learning_bound}. For $t \ge 4 e$ we have,
\[
\frac{4 n \log(t /4)}{t \eta_t} \ge \frac{4 n }{t \eta_t}
\]
Therefore,
\[
\sup_{t \ge 4 e} \frac{e^{-t \eta_t / 16 n}}{\frac{4 n \log(t /4)}{t \eta_t} } \le \sup_{t \ge 4 e} \frac{e^{-t \eta_t / 16 n}}{\frac{4 n}{t \eta_t} } = \frac{4}{e}\]
\[
\sup_{t \ge 4 e} \frac{e^{-t \eta_t / 8}}{\frac{4 n \log(t /4)}{t \eta_t}} \le \sup_{t \ge 4 e} \frac{e^{-t \eta_t / 8}}{\frac{4 n}{t \eta_t} } = \frac{2}{n \cdot e}
\]
Subsequently, we have for $t \ge 4 e$ that:
\[
\mathbb{E}[\Delta_t] \le c \cdot \frac{n \log(t)}{t \eta_t} \ , \quad c = 4 \cdot \left(1 + 6/e \right)
\]
Now we take $\eta_t = t^{-1/2} \cdot (n \log(t))^{(1+\epsilon)/2}$ and for $t \ge 4 e$ we arrive at,
\[
\mathbb{E}[\eta_t] = t^{-1/2} \cdot (n \log(t))^{(1+\epsilon)/2}\]
\[\mathbb{E}[\Delta_t] \le  c \cdot \frac{n \log(t)}{t^{1/2} \cdot (n \log(t))^{(1+\epsilon)/2}} = c \cdot t^{-1/2} (n \log(t))^{(1-\epsilon)/2}
\]
We see from an integration that $\text{Regret}(H) = O(H^{1/2} \cdot \log(H)^{(1 + \epsilon) / 2})$ which establishes the result. Since $\Sigma_t \mathbb{E}[\Delta_t] = o(\Sigma_t \mathbb{E}[\eta_t])$ this regret bound is consistent with the assumption in Theorem \ref{main_result}. 

\end{proof}

\subsection{Proof of theorem \protect \ref{pairMech}}
\label{Example_Proof}

\pairMech* 

To get the result, we need to show linear regression converges and that it converges to expected utility. Here we'll establish that the expectation of the reports does in fact equal expected utility. Later we will establish convergence of the regression to expected utility. 

\meanPair*

\begin{proof}
    Given any measurable function $u : \Omega \to [0, \infty]$ from a sample space $\Omega$ to the nonnegative real numbers, there is a sequence of nonnegative simple functions $u_t$ increasing pointwise to $u$. This means we can construct a sequence of nonnegative simple random variables $u_t$, increasing to $u$. To proceed, we look at the range set $R_t$ of the $u_t$, which will be of finite cardinality. Moreover, each element in the range will be separated with some margin $\epsilon > 0$. Index these sets with $R_t(k)$ in order taking $R_t(0)$ to be a greatest lower bound. We have,

    \[
    \sum_{k = R_t(0)}^{|R_t|} P(u_t \ge R_t(k)) = \sum_{k = R_t(0)}^{|R_t|} \sum_{m = k}^{|R_t|} P(m + \epsilon > u_t \ge m)
    \]
    \[
    = \sum_{m = R_t(0)}^{|R_t|} \sum_{k = 1}^{m} P(m + \epsilon > u_t \ge m) 
    \]
    \[
    = \sum_{m = R_t(0)}^{|R_t|} m P(m+\epsilon > k \ge m)  = \sum_{m = R_t(0)}^{|R_t|} \int_{[m,m+\epsilon)} m d \nu
    \]
    \[
    = \sum_{m = R_t(0)}^{|R_t|} \int_{[m,m+\epsilon)} u_t d \nu = \int u_t d \nu = \mathbb{E}[u_t]
    \]
    
    Our construction also implies $\mathbb{P}(u_t \ge z)$ increases to $\mathbb{P}(u \ge z)$ monotonically with respect to the sequence. It is not hard to see that this sequence of functions is measurable. Therefore, by the monotone convergence theorem,

    \[
    \mathbb{E}[u] = \lim_{n \to \infty} \int_0^{\infty} \mathbb{P}(u_t \ge z) dz = \int_0^{\infty} \mathbb{P}(u \ge z) dz
    \]

    which concludes the proof. 
\end{proof}

For the main result. It is relatively well-known that well-specified linear regression converges in its predictions.

\begin{lemma}{(\citep{gyorfi2002distribution})}
Let $\hat \mu_t$ be an empirical least squares estimator for a linear valuation function $\mu$, which may depend on the input data $w_1, \ldots, w_t$. Denote the $L_2(\nu)$ norm with respect to the probability measure $\nu$ of the data as $\Vert \cdot \Vert$. Then,
\[
\mathbb{E}[\Vert \hat \mu_t - \mu \Vert^2  ] \le c \cdot \cfrac{(\log(t) + 1) \cdot d}{t}
\]
for some constant $c$. 
\end{lemma}

Since $\nu$ is a probability measure, by the monotone property for $L_p(\nu)$ norms we have that,

\[\mathbb{E}[\Vert \hat \mu_t - \mu \Vert  ] \le \sqrt{\mathbb{E}[\Vert \hat \mu_t - \mu \Vert^2  ]} \le \sqrt{c \cdot \cfrac{(\log(t) + 1) \cdot d}{t}}\]

Recall that $\Delta_t$ is the expected maximum error of the learning algorithm over all the agents $i \in [n]$ with respect to exploration round data. A simple bound works for our purposes,

\[\Delta_t \le n \cdot \mathbb{E}[\Vert \hat \mu_t - \mu \Vert  ] = O \left( \sqrt{ \cfrac{\log(t)}{t}} \right)\]

so with $\eta_t = t^{-1/3} \cdot (n \log(t))^{(1+2 \epsilon)/3}$ we satisfy the necessary condition to invoke Theorem \ref{main_result}. 

\end{document}